\documentclass[11pt]{article}
\usepackage{amsmath,amssymb,amsthm}
\usepackage{fullpage,parskip}

\theoremstyle{plain}
\newtheorem{theorem}{Theorem}[section]
\newtheorem{lemma}[theorem]{Lemma}

\newtheorem{coro}[theorem]{Corollary}

\theoremstyle{definition}
\newtheorem{defi}[theorem]{Definition}

\title{A basic lower bound for property testing\thanks{Research supported in part by an Israel Science Foundation grant number 879/22.}}
\author{Eldar Fischer\thanks{Technion -- Israel Institute of Technology. \tt eldar@cs.technion.ac.il}}

\date{}

\begin{document}

\maketitle
\begin{abstract}\noindent
An $\epsilon$-test for any non-trivial property (one for which there are both satisfying inputs and inputs of large distance from the property) should use a number of queries that is at least inversely proportional in $\epsilon$. However, to the best of our knowledge there is no reference proof for this intuition in its full generality. Such a proof is provided here. It is written so as to not require any prior knowledge of the related literature, and in particular does not use Yao's method.
\end{abstract}

\section{Introduction and statement}\label{sec:intro}

Let us fix a constant size alphabet $\Sigma$, and for any $n\in\mathbb N$ consider the set $\Sigma^n$ of all possible inputs of size $n$. We analyze algorithms that are given random access to an input $X\in\Sigma^n$ (with the input size $n$ being known in advance to the algorithm), but in general do not run for sufficiently many steps to read the input in its entirety. Such (randomized) algorithms are defined as follows.

\begin{defi}
For an alphabet $\Sigma$, a {\em query-making algorithm $\mathcal A$ over $\Sigma^n$ with $q$ queries} is a random algorithm that is given $n$ in advance, and runs over an input $X\in\Sigma^n$ in $q$ stages as follows: At Stage $i$, the algorithm selects an index $j_i\in [n]$ based only on $X(j_1),\ldots,X(j_{i-1})$ and its internal coin tosses, and gains access to the value $X(j_i)$. After $q$ stages have concluded, the algorithm either accepts or rejects the input based on $X(j_1),\ldots,X(j_q)$ and its internal coin tosses.
\end{defi}

The above definition is for adaptive algorithms, where the query made in every stage can depend on the outcomes of the preceding stages. We also analyze algorithms for which this is not the case.

\begin{defi}
A {\em non-adaptive query-making algorithm with $q$ queries} is an algorithm $\mathcal A$ that selects $j_1,\ldots,j_q$ based only on its internal coin tosses (i.e.\ it does not gain access to the input until the querying stages are over), and then either accepts or rejects the input based on $X(j_1),\ldots,X(j_q)$ and its internal coin tosses.
\end{defi}

A property $\mathcal P$ can be any subset of $\bigcup_{n\in\mathbb N}\Sigma^n$. In general, for $q(n)=o(n)$, a query-making algorithm with $q(n)$ queries for an input $X\in\Sigma^n$ cannot solve the exact decision problem of belonging to $\mathcal P$ for most reasonable properties $\mathcal P$. However, such algorithms can at times solve an approximate task where for some inputs any answer is acceptable.

A property-testing algorithm is a query-making algorithm that can distinguish inputs satisfying a certain property from inputs that are far from the property. To define them (in the ``classical'' testing model), we first define what it means to be far.

\begin{defi}
The {\em normalized Hamming distance} between two same-size inputs $X,Y\in\Sigma^n$ is defined as the number of differences between the inputs divided by their input size, that is $d(X,Y)=|\{j:X(j)\neq Y(j)\}|/n$. For a property $\mathcal P$ we define its distance from $X$ by the minimum, $d(X,\mathcal P)=\min_{Y\in\mathcal P\cap\Sigma^n}d(X,Y)$.
\end{defi}

The following is the formal definition of classical property testing. Such algorithms were first defined and investigated in \cite{BLR:org} and \cite{RS:org}. The constant $\frac23$ below is arbitrary and can be replaced with any other constant larger than $\frac12$.

\begin{defi}
Given $\epsilon>0$, an {\em $\epsilon$-testing algorithm} for a property $\mathcal P$ is a query-making algorithm (adaptive or non-adaptive), so that every $X\in\Sigma^n\cap \mathcal P$ is accepted with probability at least $\frac23$, and every $X\in\Sigma^n$ for which $d(X,\mathcal P)\geq\epsilon$ is rejected with probability at least $\frac23$.
\end{defi}

In the above definition, any acceptance probability is allowed for an input $X$ for which $X\notin \mathcal P$ and $d(X,P)<\epsilon$. Since its inception, property testing as well as variants with other definitions of distances and allowable queries have been vigorously investigated. See \cite{G:book} for a thorough tutorial.

A straightforward intuition is that for ``non-trivial'' properties, i.e.\ properties that admit both satisfying inputs and inputs that are far from the property (in the sense of having a constant lower bound on their distance from the property), the $\epsilon$-testing task would require a number of queries that is at least inversely proportional to $\epsilon$. However, to the best of our knowledge no formal proof of this intuition for general properties has ever been published.

In \cite{BG:nont} this intuition is proved for all binary properties that are not too dense. Specifically they prove that all properties that contain at least one but not more than $2^{n-\Omega(n)}$ satisfying instances from $\{0,1\}^n$ require $\Omega(1/\epsilon)$ many queries for an $\epsilon$-test. A weaker statement, but without the non-density requirement, is proved in \cite{GKS:estr} for properties in the dense graph model (first defined in \cite{GGR:dense}). These properties are required to be closed under the group of possible graph isomorphisms. The bound in \cite{GKS:estr} is $\Omega(1/\sqrt{\epsilon})$, and its proof uses the conversion of \cite{GT:three} of any $\epsilon$-test under this model (at quadratic cost) into a non-adaptive one for which every possible index is queried with equal probability.

We prove in Section \ref{sec:proof} the following statement, which implies the general intuition for all non-trivial properties. Its proof here is based on first principles only, and is written on purpose to not require the customary use of Yao's principle, even that using it would have shortened the proof.

\begin{theorem}[Meticulous statement]\label{thm:met}
If a property $\mathcal P\subset\Sigma^n$ (for any alphabet $\Sigma$) admits both a satisfying input $S\in\Sigma^n\cap\mathcal P$ and an input $U\in\Sigma^n$ for which $d(U,\mathcal P)\geq k/n$, then for any $1\leq l<k$, an $(l/n)$-test for $\mathcal P$ requires at least $\Omega(k/l)$ many queries to meet its guarantees.
\end{theorem}

Its immediate corollary is the following more comprehensible statement.

\begin{coro}[Straightforward statement]
If for a fixed $\alpha>0$ a property $\mathcal{P}$ over any alphabet $\Sigma$ admits, for infinitely many $n$, both a satisfying input $S_n\in\Sigma^n\cap\mathcal P$ and an input $U_n\in\Sigma^n$ for which $d(U_n,\mathcal P)\geq\alpha$, then for any $\epsilon<\alpha$ and infinitely many values of $n$, an $\epsilon$-test for $\mathcal P$ requires at least $\Omega(\alpha/\epsilon)$ many queries to meet its guarantees.
\end{coro}

\begin{proof}
We consider any $n>\max\{1/\epsilon,1/(\alpha-\epsilon)\}$ for which there is an input $S_n\in\Sigma^n\cap\mathcal P$ and an input $U_n\in\Sigma^n$ for which $d(U_n,\mathcal P)\geq\alpha$. For $0<\epsilon<\alpha$ and $n$ as above, denoting $k=\lceil\alpha n\rceil$ and $l=\lceil\epsilon n\rceil$, we deploy Theorem \ref{thm:met}, noting that $1\leq l<k$ (by the choice of $n$) and $k/l=\Omega(\alpha/\epsilon)$ in this case.
\end{proof}

\section{The main proof}\label{sec:proof}

Using the notation of the statement of Theorem \ref{thm:met}, we set $k'=d(U,\Sigma^n\cap\mathcal P)\cdot n$, noting that this is an integer satisfying $k'\geq k$. We also let $V\in\Sigma^n\cap\mathcal P$ to be the satisfying input that is closest to $U$ (i.e.\ of distance $k'/n$), and let $D=\{j\in [n]:U(j)\neq V(j)\}$ be the set of indexes where $U$ and $V$ differ. In particular, $|D|=k'$.

For every set $A\subseteq D$, we define by $U_A\in\Sigma^n$ the input resulting from replacing the values of $V$ that relate to $A$ with the values of $U$. That is, for $j\in A$ we define $U_A(j)=U(j)$, while for $j\notin A$ we define $U_A(j)=V(j)$. Note that in particular $U_{\emptyset}=V$ while $U_D=U$.

\begin{lemma}\label{lem:graded}
For every $A\subseteq D$ we have $d(U_A,\mathcal P)=|A|/n$.
\end{lemma}

\begin{proof}
On one hand we have $d(U_A,\mathcal P)\leq d(U_A,V)=|A|/n$. On the other hand by the triangle inequality we have $d(U_A,\mathcal P)\geq d(U,\mathcal P)-d(U_A,U)=|D|/n-(|D|-|A|)/n=|A|/n$.
\end{proof}

Using $U$ and $V$ above, we define a task that is weaker than a test but is easier to bound against.

\begin{defi}
For every $0\leq l\leq k'$, a {\em $(U,V,l)$-distinguisher} is a query-making algorithm that accepts with probability at least $\frac23$ the input $V$, and rejects with probability at least $\frac23$ any input from the family $\mathcal U_l=\{U_A:A\subseteq D \wedge |A|=l\}$. All other inputs can be accepted with any probability.
\end{defi}

Note that by Lemma \ref{lem:graded} every $(l/n)$-test for $\mathcal P$ is in particular a $(U,V,l)$-distinguisher. We show that without loss of generality, algorithms for the latter task can be assumed to be quite simple

\begin{lemma}\label{lem:in}
Without loss of generality, a $(U,V,l)$-distinguisher does not take queries outside $D$.
\end{lemma}

\begin{proof}
Given a $(U,V,l)$-distinguisher $\mathcal A$, we construct a $(U,V,l)$-distinguisher $\mathcal A'$ that does not take queries outside $D$ (and never takes more queries than $\mathcal A$) by the following: We run $\mathcal A$, and whenever at some stage $i$ the algorithm is about to take a query $X(j_i)$ from an index $j_i\notin D$, we do not take it and instead continue running as if the answer was $U(j_i)=V(j_i)$.

Note that by the definition of a $(U,V,l)$-distinguisher, if $X(j)\neq U(j)$ for some $j\in [n]\setminus D$ then the output of $\mathcal A'$ over $X$ does not matter, while in the case where $X(j)=U(j)$ for all $j\in [n]\setminus D$ both $\mathcal A$ and $\mathcal A'$ behave the same.
\end{proof}

\begin{lemma}\label{lem:nad}
Without loss of generality, a $(U,V,l)$-distinguisher does not take queries outside $D$ and is additionally non-adaptive.
\end{lemma}

\begin{proof}
Given a $(U,V,l)$-distinguisher $\mathcal A$, we first construct a $(U,V,l)$-distinguisher $\mathcal A'$ that does not take queries outside $D$ as per Lemma \ref{lem:in}. We then construct a $(U,V,l)$-distinguisher $\mathcal A''$ as follows: We follow the run of $\mathcal A'$, but in every stage $i$, for the calculation of where to take the query $j_i$, instead of the values $X(j_1),\ldots,X(j_{i-1})$ we use the values $V(j_1),\ldots,V(j_{i-1})$. After all queries have been taken, we check whether there is any $i\in [q]$ for which $X(j_i)\neq V(j_i)$. If there is such an $i$ we reject the input, and otherwise we follow the decision of $\mathcal A'$ with respect to $X(j_1),\ldots,X(j_q)$.

We first note that if $X(j_i)\neq V(j_i)$ for some $i\in [q]$, then it cannot be the case that $X=V$, and hence rejecting the input in such a case is never a wrong decision (it is either a correct decision, or the situation is such that the output of the algorithm does not matter). We next note that conditioned on picking $j_1,\ldots,j_q$ such that $X(j_i)=V(j_i)$ for all $i\in [q]$, both $\mathcal A'$ and $\mathcal A''$ behave the same over $X$. Finally we note that $\mathcal A''$ is non-adaptive, since the decisions about which queries to take do not depend on $X$ at all (they only depend on the fixed in advance $V$).
\end{proof}

The following lemma immediately implies and concludes the proof of Theorem \ref{thm:met}.

\begin{lemma}[implying Theorem \ref{thm:met}]
Any $(U,V,l)$-distinguisher, and hence any $l/n$-testing algorithm for $\mathcal P$, requires at least $|D|/3l=\Omega(k/l)$ queries to meet its guarantees.
\end{lemma}

\begin{proof}
Given a $(U,V,l)$-distinguisher $\mathcal A$, we first use Lemma \ref{lem:nad} and assume without loss of generality that $\mathcal A$ is non-adaptive and only takes queries in $D$. Let $q$ be its number of queries.

For every $j\in D$ we now denote by $E_j$ the event that the query $X(j)$ was taken from the input at any stage of the algorithm (that is, that $j_i=j$ for some $i\in [q]$), and define $p_j=\Pr[E_j]$. Since at any run of the algorithm at most $q$ of the events $\{E_j:j\in [q]\}$ happen, we have $\sum_{j\in D}p_j\leq q$.

We let $D'\subset D$ denote the set of the $l$ indexes in $D$ with the smallest (not necessarily distinct) values of $p_j$. We have then $\sum_{j\in D'}p_j\leq ql/|D|$. Let us analyze the run of $\mathcal A$ when given $X=U_{D'}$ as input, in comparison to its run when given $X=V$ as input. Denoting by $E_{D'}$ the union event $\bigcup_{j\in D'}E_j$, we obtain by the union bound $\Pr[E_{D'}]\leq\sum_{j\in D'}\Pr[E_j]\leq ql/|D|$.

The probability of $\mathcal A$ to accept the input, conditioned on $E_{D'}$ not happening, is the same for $U_{D'}$ and $V$, since conditioned on $\neg E_{D'}$ we have $U_{D'}(j_i)=V(j_i)$ for all $i\in [q]$. Hence, the unconditioned acceptance probabilities of $U_{D'}$ and $V$ differ by not more than $\Pr[E_{D'}]\leq ql/|D|$. For $q<|D|/3l$ this is less than $\frac13$. On the other hand, a $(U,V,l)$-distinguisher should accept $V$ with probability at least $\frac23$, while $U_{D'}\in \mathcal U_l$ should be accepted with probability not exceeding $\frac13$, a contradiction.
\end{proof}

\bibliographystyle{amsplain}
\bibliography{lb}

\end{document}